\documentclass[
final
]{dmtcs-episciences}

\usepackage[utf8]{inputenc}
\usepackage{subfigure}

\usepackage[usenames,dvipsnames]{xcolor}
\usepackage{makeidx} 
\usepackage{algorithm}
\usepackage{algorithmic}
\usepackage{graphicx,tipa}
\usepackage{arcs,lmodern,fix-cm}
\usepackage{times}
\usepackage{amsfonts,latexsym,graphicx,epsfig,amssymb,color}
\usepackage{mathdots,amsmath,amsthm,amstext,setspace}
\usepackage{amsmath,amstext,amsthm,url,setspace, enumerate}
\usepackage{slashbox,multirow}
\usepackage{rotating}
\usepackage{verbatim}

\newcommand{\tarc}{\mbox{\large$\frown$}}
\newcommand{\arcc}[1]{\stackrel{\tarc}{#1}}

\floatname{algorithm}{Algorithm}

\newcommand{\barr}[1]{\overline{#1}}

\newcommand{\sinn}[1]{\sin \left({#1}\right)}

\newcommand{\coss}[1]{\cos \left({#1}\right)}

\newcommand{\ignore}[1]{}
\newcommand{\alg}[1]{$\mathcal{A}_{#1}$}

\newcommand{\ro}[1]{$R_{#1}$}

\def\dist{\alpha}
\def \boundnw {\pi-{\dist}/2+3\sinn{{\dist}/2}}
\def \boundw {\pi-{\dist}+4\sinn{{\dist}/2}}

\def\tew{2-TE$_w$}
\def\ten{2-TE$_{f2f}$}

\newtheorem{theorem}{Theorem}[section]

\newtheorem{lemma}[theorem]{Lemma}

\theoremstyle{definition}

\usepackage[round]{natbib}

\author{
Konstantinos Georgiou\affiliationmark{1}
\and 
George Karakostas\affiliationmark{2}
\and 
Evangelos Kranakis\affiliationmark{3}
}
\title[Search-and-Fetch with 2 Robots on a Disk]{Search-and-Fetch with 2 Robots on a Disk: \\
Wireless and Face-to-Face Communication Models\thanks{Research supported in part by NSERC of Canada.}}
\affiliation{
  Dept. of Mathematics, 
Ryerson University, 
Toronto, ON, Canada\\
Dept. of Computing \& Software, 
McMaster University, 
Hamilton ON, Canada\\
School of Computer Science, Carleton University, Ottawa ON, Canada}
\received{2018-10-12}
\accepted{2019-5-10}
\begin{document}
\publicationdetails{21}{2019}{3}{20}{4884}
\maketitle
\begin{abstract}
We initiate the study of a problem on {\em searching and fetching}, motivated by real-life surveillance and search-and-rescue operations where unmanned vehicles, \textit{e.g.} drones, search for victims in areas of a disaster.
In \emph{treasure-evacuation}, we are interested in designing algorithms that minimize the time it takes for a treasure (a victim) to be discovered and brought (fetched) to the exit (shelter) by any of two robots (rescuers) which are performing in a distributed environment
(the case of searching and fetching with 1 robot has been previously considered).
The communication protocol between the robots is either {\em wireless}, where information is shared at any time, or {\em face-to-face}, where information can be shared only if the robots meet. 
For both models we obtain upper bounds for fetching the treasure to the exit. 
Our algorithms make explicit use of the distance between the treasure and the exit, which is assumed to be known in advance, showing this way how partial information of the unknown input can be beneficial. 
Our main technical contribution pertains to the face-to-face model. More specifically, we demonstrate how robots can exchange information without meeting, effectively achieving a highly efficient treasure-evacuation protocol which is minimally affected by the lack of distant communication. 
Finally, we complement our positive results above by providing a lower bound in the face-to-face model.
\end{abstract}

\section{Introduction}
We introduce the study of a distributed problem on {\em searching and fetching}
called {\em treasure evacuation}. Two robots are placed
at the center of a unit disk, while an exit and a treasure 
lie at unknown positions on the perimeter of the disk.
Robots search with maximum speed $1$, and they detect a point of interest (either the treasure or the exit) only if they pass over it. The exit is immobile, while the treasure can be carried by any of the robots.
The goal of the search is for at least one of the robots
to bring (fetch) the treasure to the exit, \textit{i.e.} evacuate the treasure, in the minimum possible completion time. 
The robots do not have to evacuate, they only need to co-operate, possibly by sharing information, so as to learn the locations of the points of interest and bring the treasure to the exit. 
Contrary to previous work, this is the first time an explicit ordering on the
tasks to be performed is imposed on two moving robots. This makes the
problem inherently different in nature and more difficult than similarly looking results.

Special to our problem is also the underlying advice-model we consider, \textit{i.e.} that even though the locations of the exit and the treasure are unknown, their distance is known and is considered part of the input. 
Interestingly, 
finding an optimal algorithm turns out to be a challenging task even when the robots have this knowledge. 
We propose treasure-evacuation protocols
in two communication models. In the {\em wireless} model robots exchange information instantaneously and at will, while in the {\em face-to-face} model information can be exchanged only if the robots meet. We aim at incorporating the knowledge of the exact distance between the exit and the treasure into our algorithm designs. We offer algorithmic techniques such as planning ahead, timing according to the explicit task ordering, and retrieval of unknown information through inference and not communication.

Our problem is motivated by real-life surveillance and search-and-rescue operations where unmanned vehicles, \textit{e.g.} drones, search for victims in areas of a disaster. 
Indeed, consider a group of rescuer-mobile-agents (robots), initially located strategically in a central position of a domain. 
When alarm is triggered and a distress signal is received, robots need to locate a victim (the treasure) and bring her to safety (the exit). Our problem shares similarities also with classic and well-studied cops-and-robbers games;  robots rest at a central position of a domain (say, in the centre of a disk as in our setup) till an alarm is triggered by some ``robber'' (the treasure in our case). Then, robots need to locate the stationary robber and subsequently bring him to jail (the exit). 
Interestingly, search-and-fetch type problems resemble also situations that abound in fauna, where animals hunt for prey which is then brought to some designated area, \textit{e.g.} back to the lair. As such, further investigation of similar problems will have applications to real-life rescue operations, as well as to the understanding of animal behavior, as it is common in all search problems.

\subsection{Problem Definition \& Contributions}

A treasure and an exit are located at unknown positions on the perimeter of a unit-disk and at arc distance $\dist$ (in what follows all distances will be arc-distances, unless specified otherwise). Robots, denoted by \ro1, \ro2,  start from the center of the disk, and can move anywhere on the disk at constant speed 1. Each of the robots detects the treasure or the exit only if its trajectory passes over that point on the disk. Once detected, the treasure can be carried by a robot at the same speed. 
We refer to the task of bringing the treasure to the exit as \emph{treasure-evacuation}. We use
the abbreviations $T, E$ for the treasure and the exit, respectively.
For convenience, in the sequel we will refer to the locations
of the exit and the treasure as {\em points of interest (PoI)}.
For a PoI $I$ on the perimeter of the disk, we also write $I=E$ ($I=T$) to indicate that the exit (treasure) lies in point $I$.
For a point $B$, we also write $B=null$ to denote the event that neither the treasure nor the exit is placed on $B$.

We focus on the following online variants of treasure-evacuation with 2 robots, where the exact distance $\dist$ between $T,E$ is known, but not their positions: \\
- In \textbf{\tew} (Section~\ref{sec: w-model}), information between robots is shared continuously in the time horizon, \textit{i.e.} messages between them are exchanged instantaneously and at will with no restrictions and no additional cost or delays. \\
- In \textbf{\ten} (Section~\ref{sec: f2f-model}), the communication protocol between the robots is face-to-face (non-wireless)---abbreviated F2F (or f2f), where information can be exchanged only if the robots meet at the same point anywhere.

Part of our contribution is that we demonstrate how robots can utilize the knowledge of the arc-distance $\dist$ between the points of interest. We propose protocols that induce worst case evacuation time $1+ \boundw$ 
for the wireless model and 
$1+\pi-\dist/2+3\sinn{\dist/2}$
for the face-to-face model. 
The worst case cost for the two problems becomes $1+2 \sqrt{3}+\frac{\pi }{3}\approx 5.5113$ (when $\dist= \frac{2 \pi }{3}\approx2.0944$)
and $1+2 \sqrt{2}+\pi -\sec ^{-1}(3)\approx 5.73906$ (when $\dist=2 \sec ^{-1}(3)\approx2.46192$), respectively.
The upper bound in the face-to-face model, which is our main contribution, is the result of a non-intuitive evacuation protocol that allows robots to exchange information about the topology without meeting, effectively bypassing their inability to communicate from distance.
Note that our results induce upper bounds with respect to competitive analysis as well. Indeed, the optimal solution, given that the input is known, equals $opt_\dist=1+2\sinn{\dist/2}$, hence the competitive ratio we achieve, for fixed $\dist$, can be computed by scaling the worst case performance we achieve by $opt_\dist$. In both cases, the worst case competitive ratio becomes $1+\pi$, for $\dist\rightarrow 0$. 
Finally, we complement our results above by showing that any algorithm in the face-to-face model needs time at least $1+\pi/3+4\sinn{\alpha/2}$, if $\alpha \in [0,2\pi/3]$ and at least $1+\pi/3+2\sinn{\alpha}+2\sinn{\alpha/2}$, if $\alpha \in [2\pi/3,\pi]$.
A graphical comparison of our results can be seen in Figure~\ref{fig: compUBLB}.
\begin{figure}[!ht]
                \centering
                \includegraphics[scale=0.5]{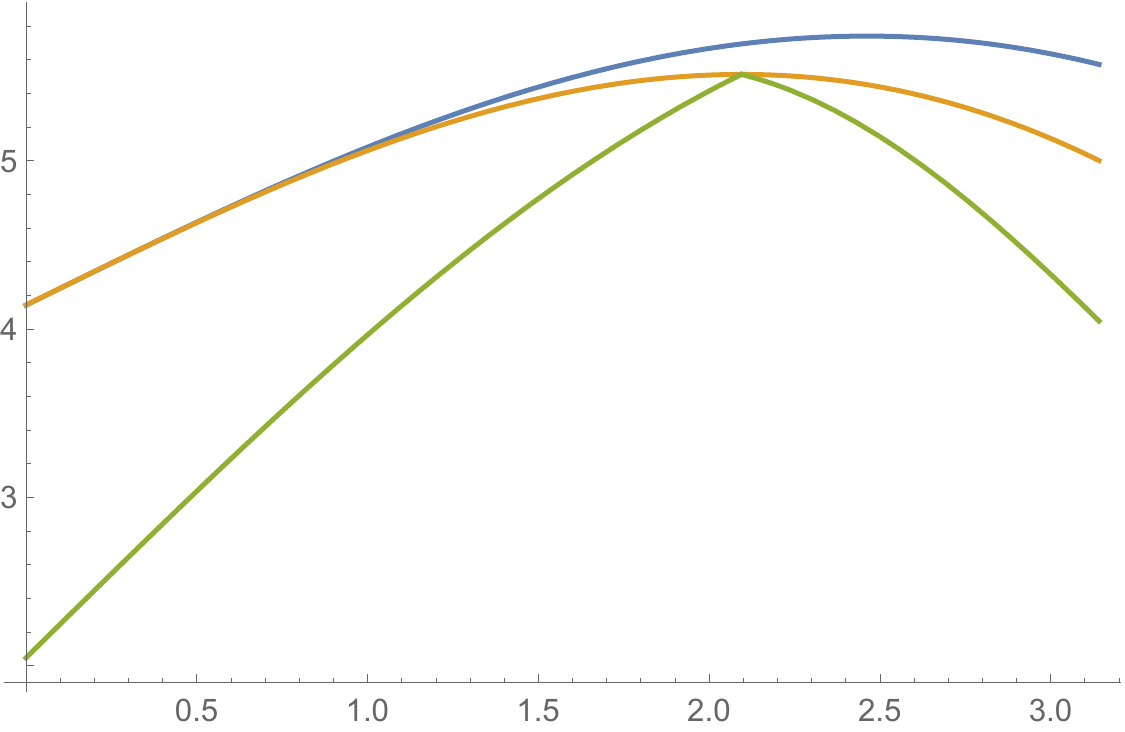}
                \caption{Comparison between the performance of the wireless algorithm (yellow curve), the performance of the f2f algorithm (blue curve) and the provided lower bound (green curve) depicted on the vertical axis, as a function of $\dist$ (horizontal axis).}
                \label{fig: compUBLB}
\end{figure}

\subsection{Problem and Model Motivation}

From a technical perspective, our communication models are inspired by the recent works on evacuation problems~\cite{CGGKMP,CGKNOV,CKKNOS}. Notably, the associated search problems are inherently different than our problem which is closer in nature to search-type, treasure-hunt, and exploration problems. Also, our mathematical model features (a) a distributed setting (b) with objective to minimize time, and (c) where different communication models are contrasted. None among (a),(b),(c) are well understood for search games, and, to the best of our knowledge, they have not been studied before in this combination. 

Specific to the problem we study are the number of robots (2 and not arbitrarily many - though our results easily extend to swarms of robots), the domain (disk), and the fact the robots have some knowledge about the PoI. Although extending our results to more generic situations is interesting in its own right, the nature of the resulting problems would require a significantly different algorithmic approach. Indeed, our main goal is to study how limitations in communication affect efficiency, which is best demonstrated when the available number of robots, and hence computation power, is as small as possible, \textit{i.e.} for two robots. In fact, it is easy to extend our algorithms for the $n$-robot case. 

Notably, search-and-fetch problems are challenging even for 1 robot as demonstrated in~\cite{GKKa16}. 
In particular, the work of \cite{GKKa16} implies that establishing provably optimal evacuation protocols for 2-robots is a difficult task, even when the domain is the disk.
Indeed, the best-known trajectory for 1 robot in \cite{GKKa16}, which is also conjectured to be optimal, exhibits delicate ``jumps'' that effectively save an almost negligible amount of the termination time, still they improve upon the naive approach. 
 Nevertheless, we view the disc domain that we study as natural. Indeed, a basic setup in search-and-rescue operations is that rescuer-robots inhabit in a base-station, and they stay inactive till they receive a distress signal. As it is common in real-life situations, the signal may only reveal partial information about the location of a victim, \textit{e.g.} its distance from the base-station, along with the distance between the points. When there are more than one PoI to be located, this kind of information suggests that the points lie anywhere on co-centric circles. When the points are equidistant from the base-station, robots need only consider a disk, as it is the case in our problem. We believe that with enough technical and tedious work, our results can also extend to non-equidistant points, however the algorithmic significance of the proposed distributed solutions may be lost in the technicalities. 

Finally, specific to our problem is the underlying advice model. Indeed, robots have access to partial information (the exact distance of the hidden objects) about the unknown input (the exact locations of the hidden objects). 
Partial knowledge of the input is interesting due to efficiency-information tradeoffs that are naturally induced by online problems, 
and which are commonly studied in competitive analysis, \textit{e.g.} see~\cite{HIKL99,CGDKM1618} and \cite{GKKa16}.
In our search-and-fetch problem, the partial information of the distance of the hidden objects demonstrates that robots with primitive communication capabilities are in fact not much less powerful than in the wireless model.
The reader may also view this piece of advice as an algorithmic challenge in order to bypass the uncertainty regarding the locations of the PoI. 
Notably, our algorithms adapt strategies as a function of the distance of the PoI, 
trying to follow protocols that would allow them to detect the actual positions of the points without necessarily visiting them. As an easy example, note that if a robot has explored already a contiguous arc of length $\alpha+\epsilon$, the discovery of a PoI reveals the location of the other $\alpha$-arc distant away PoI (our algorithm makes use of distance $\alpha$ in a much more sophisticated way). As a result, had we assumed that distances are unknown, robots may not be able to deduce such important information about the topology using partial exploration, and the problem would require an inherently different algorithmic approach.

\subsection{Related Work}
\label{sec: related work}

Traditional search is concerned with finding an object
with specified properties within a search space.
Searching in the context of computational problems is usually more challenging
especially when the environment is unknown to the searcher(s) (see \cite{ahlswede1987search,alpern2003theory,stone1975theory}).
This is particularly evident in the context of robotics 
whereby exploration is taking place within a given geometric 
domain by a group of autonomous but communicating robots. 
The ultimate goal is to design an algorithm so as to 
accomplish the requirements of 
the search (usually locating a target of unknown a priori
position) while at the same time obeying the computational and geographical
constraints. The
input robot configuration must also accomplish the task in the
minimum possible amount of time~\cite{berman1998line}. 

Search has a long history. There is extensive and varied research
and several models have been proposed and
investigated in the mathematical and theoretical
computer science literature with particular emphasis
on probabilistic search~\cite{stone1975theory},
game theoretic applications~\cite{alpern2003theory},
cops and robbers~\cite{anthony2011game},
classical pursuit and evasion~\cite{nahin2012chases},
search problems as related to group testing~\cite{ahlswede1987search},
searching a graph~\cite{koutsoupias1996searching},
and many more. 
A survey of related search and pursuit evasion problems
can be found in \cite{chung2011search}. In pursuit-evasion,
pursuers want to capture evaders who try to avoid capture.
Examples include {\em Cops and Robbers} (whereby the cops
try to capture the robbers by moving along the vertices 
of a graph), {\em Lion and Man} (a geometric version of
cops and robbers where a lion is to capture a man in either
continuous or discrete time), etc.
Searching for a motionless point target has some similarities with the lost at sea problem, \cite{G61,I67}, the cow-path problem \cite{beck1964linear,bellman1963optimal}, and with the plane searching problem \cite{BS95}. This last paper also
introduced the ``instantaneous contact model'', which is 
referred to as wireless model
in our paper.
When  the mobile robots do not know the geometric environment in advance then researchers  are concerned with exploring \cite{AH00,AKS02,DKP91,HIKK01}. Coordinating the exploration of a team of robots is a main theme in the robotics community \cite{B05,T01,Y98} and often this is combined with the mapping of the terrain and the position of the robots within it \cite{K94,PY}.

Evacuation for grid polygons has been studied in \cite{FGK10}
from the perspective of constructing centralized evacuation plans, resulting in the fastest possible evacuation from the rectilinear environment.
There are certain similarities of our problem to
the well-known evacuation problem on an infinite line (see \cite{baezayates1993searching} and the recent~\cite{Groupsearch})
in that the search is for an unknown target. However, in this work the adversary has limited possibilities since search is on a line.
Additional research and variants on this problem
can be found in~\cite{demaine2006online} (on searching
with turn costs), \cite{kao1996searching} (randomized algorithm for the cow-path problem),
\cite{kao1998optimal} (hybrid algorithms), 
\cite{BampasCGIKKP16} (searching with different speeds), 
and many more.

A setting similar to ours is presented in the recent
works~\cite{CGGKMP,CGDKM1618,CGKNOV,CKKNOS,LMS16,BLLSW17,CGS18}
where
algorithms are presented in the wireless and non-wireless
(or face-to-face) communication models for the evacuation 
of a team of robots. The ``search domain'' in~\cite{CGGKMP,CGDKM1618,CGKNOV,CGS18}
is a unit circle (while in \cite{CKKNOS} the search
domain is a triangle or square), however, 
unlike our search problem,
in these papers
all the robots are required to evacuate from an unknown exit 
on the perimeter. Moreover, in none of these papers is there
a treasure to be fetched to the exit. 
Finally, in some more recent papers~\cite{Priority-3-servants-18, asymtotic18}, Czyzowicz et al. considered the problem of evacuating a distinguished (as in our case) mobile (unlike our case) robot. 

The problem we consider is a direct generalization of the search-and-fetch problem of~\cite{GKKa16}
with 1 robot. Unlike in our problem, searching only with 1 robot requires an almost orthogonal approach in order to improve upon the naive strategies. Indeed, the best known trajectory of~\cite{GKKa16} employs alternating moves along chords and arcs whose lengths depend on the distance of the hidden items. The induced gain is comparable to the difference between the length of an arc and its corresponding chord, and even though this quantity is not significant, it is conjectured that it is indeed the best one can achieve. In contrast, when searching with 2 robots there are more significant gains by carefully synchronizing the moves of the robots.

Our work 
is also an attempt to analyze theoretically search-and-fetch problems that have been studied by the robotics community since the 90's, \textit{e.g.} see~\cite{JWE97}. A scenario similar to ours (but only for 1 robot) has been introduced by Alpern in~\cite{Alpern11}, where the domain was discrete (a tree) and the approach/analysis resembled that of standard search-type problems~\cite{alpern2003theory}. In contrast, our problem is of distributed nature, and our focus is to demonstrate how robots' communication affects efficiency under the assumption that partial information about the input is known. 

\section{Wireless Model}
\label{sec: w-model}


As a warm-up we present in this section an upper bound for the wireless model, which will also serve as a reference for the more challenging face-to-face model. 
The algorithmic solution we propose is simple and it is meant to help the reader familiarize with basic evacuation trajectories that will be used in our main contribution pertaining to the face-to-face model. 
\begin{theorem}\label{thm: wireless upper bound}
For every $\dist\in [0,\pi]$, problem \tew\ can be solved in time $1+\boundw$. 
\end{theorem}
To prove Theorem~\ref{thm: wireless upper bound}, we propose Algorithm~\ref{def: wireless algo} that achieves the promised bound. 
Intuitively, our algorithm follows a greedy like approach, adapting its strategy as a function of the distance $\dist$ of the PoI. If $\dist$ is small enough, then the two robots  move together to an arbitrary point on the disk and start exploring in opposing directions. 
Otherwise the two robots move to two antipodal points and start exploring in the same direction. 
Exploration continues till a PoI is found. 
When that happens, the robot that can pick up the treasure and fetch it to the exit in the fastest time (if all locations have been revealed) does so, otherwise remaining locations are tried exhaustively. Detailed descriptions of the evacuation protocol can be seen in Algorithm~\ref{def: wireless algo}, complemented by Figure~\ref{fig: general pic 3 dist - wireless}.

Noticeably, the performance analysis we give is tight, meaning that for every $\alpha\geq 0$, there are configurations (placements of the PoI) for which the performance of the algorithm is exactly $1+\boundw$. Most importantly, the performance analysis makes explicit that two specific naive algorithms that do not adapt strategies together with $\alpha$ are bound to perform strictly worse than our upper bound. Also, the achieved upper bound should be contrasted to the upper bound for the face-to-face model (which is achieved by a much more involved algorithm), which at the same time is only $\alpha/2-\sinn{\alpha/2}$ more costly than the bound we show in the wireless model. 

Algorithm~\ref{def: wireless algo} takes advantage of the fact that robots can communicate to each other wirelessly. This also implies that lack of message transmission is effectively another method of information exchange. In what follows point $A$ will always be the starting point of \ro2, and $A'$ denotes its antipodal point. For the sake of the analysis and w.l.o.g. we will assume that \ro2 is the one that first finds a PoI $I=\{E,T\}$, say at time $x:=\arcc{AI}$. We call $B,C$ the points that are at clockwise and counter-clockwise arc-distance $\alpha$ from $I$ respectively. 
Figure~\ref{fig: general pic 3 dist - wireless} depicts the PoI encountered.

\begin{figure}[!ht]
                \centering
                \includegraphics[scale=0.5]{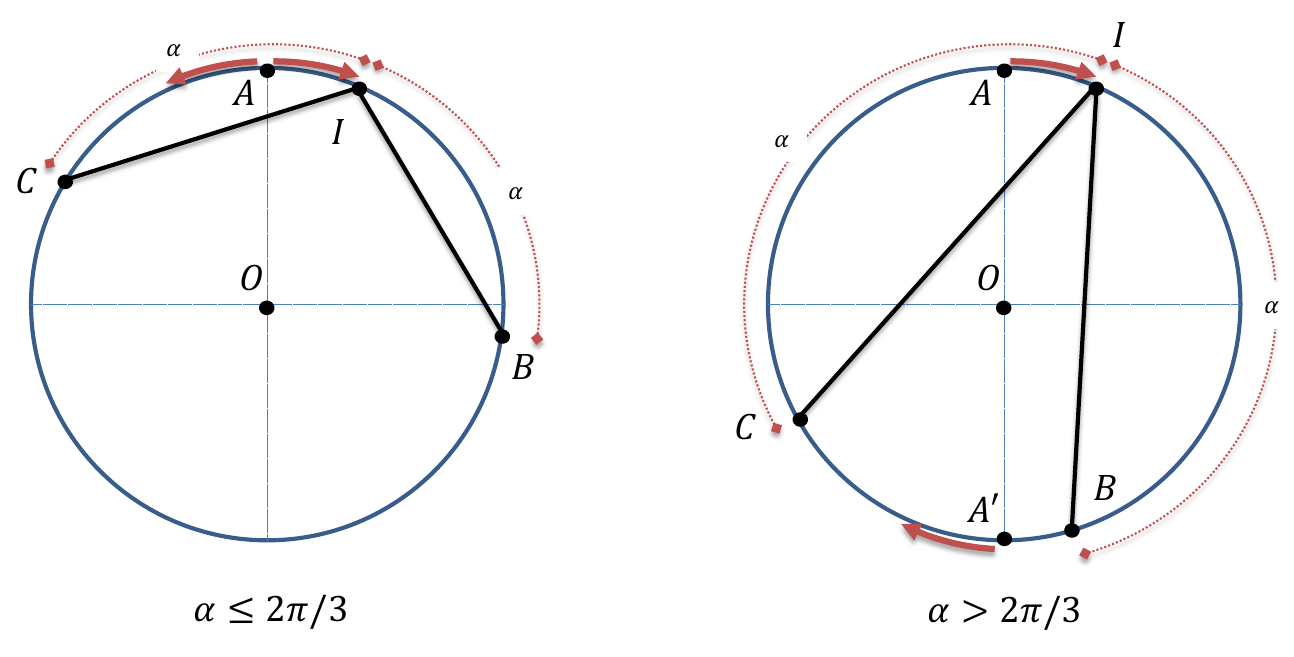}
                \caption{The points of interest for our Algorithm \ref{def: wireless algo}.}
                \label{fig: general pic 3 dist - wireless}
\end{figure}

The description of Algorithm~\ref{def: wireless algo} is from the perspective of the robot that finds first a PoI, that we always assume is \ro2. Next we assume that the finding of any PoI is instantaneously observed by the two robots. Also, if at any moment, the positions of the PoI are learned by the two robots, then the robots attempt a ``confident evacuation'' using the shortest possible trajectory. This means for example that if the treasure is not picked up by any robot, then the two robots will compete in order to pick it up and return it to the exit, moving in the interior of the disk. 

\begin{algorithm}
\caption{Wireless Algorithm}
\label{def: wireless algo}
\begin{description}
\itemsep=-\parsep
\item
{\bf Step 1.} If ${\dist}\leq 2\pi/3$, the two robots move together to an arbitrary point on the ring and start exploring in opposing directions, else they move to arbitrary antipodal points $A,A'$ on the cycle and start moving in the same direction.
\item
{\bf Step 2.} Let $I$ be the first PoI discovered by \ro2, at time $x:=\arcc{AI}$. Let $B,C$ be the points that are at clockwise and counter-clockwise arc-distance $\alpha$ from $I$, respectively.
\item
{\bf Step 3.} If $x\geq {\dist}/2$ then robots learn that the other PoI is in $B$, 
else \ro2 moves to $B$, \ro1 moves to $C$.
\item
{\bf Step 4.} Evacuate
\end{description}
\end{algorithm}

Correctness of Algorithm \ref{def: wireless algo} is straightforward, since the two robots follow a ``greedy-like evacuation protocol'' (still, they use different starting points depending on the value of $\dist$). Also, the performance analysis of the algorithm, effectively proving Theorem~\ref{thm: wireless upper bound}, 
is a matter of a straightforward case-analysis. We note that our worst-case analysis is tight, in that for every $\alpha\geq0$ there exist configurations in which the performance of Algorithm \ref{def: wireless algo} is exactly as promised by Theorem~\ref{thm: wireless upper bound}. 
Moreover, we may assume that $\dist>0$ as otherwise the problem is solved when one PoI is found. 

Note that our algorithm performs differently when ${\dist}\leq 2\pi/3$ and when ${\dist}>2\pi/3$. 
Let $x:=\arcc{AI}$ be the time that \ro2 has spent searching till first PoI $I$ is discovered. 
Then it must be that $x\leq {\dist}/2$ and $x\leq \pi-\dist$ for the cases ${\dist}\leq 2\pi/3$ and ${\dist}>2\pi/3$ respectively (see also Figure~\ref{fig: general pic 3 dist - wireless}). This will be used explicitly in the proof of the next two lemmata. We also assume that \ro2 always moves clockwise starting from point $A$. \ro1 either moves counter-clockwise starting from $A$, if ${\dist}\leq 2\pi/3$, or it moves clockwise starting from the antipodal point $A'$ of $A$, if ${\dist}>2\pi/3$. In every case, the two robots move along the perimeter of the disk till time $x$ when \ro2 transmits the message that it found a PoI.

The performance of Algorithm~\ref{def: wireless algo} is described in the next two lemmata which admit proofs by case analyses. Each of them examines the relative position of the starting point of robot \ro2 (which finds a PoI first) and the two PoI. 

\begin{lemma}\label{lem: wireless Performance when A is in}
Let $A$ be the starting point of \ro2 which is the first to discover a PoI $I$. Let also the other PoI be at $C$, where $\arcc{CI}=\alpha$. If $A$ lies in the arc $\arcc{CI}$, then the performance of Algorithm \ref{def: wireless algo} is $1+\boundw$, for all $\dist\in [0,\pi]$. 
\end{lemma}

\begin{proof}
For the case analysis, we refer to Figure~\ref{fig: wireless-algo A in}. Note that robots spend time 1 to reach the periphery of the disk. Below we calculate the remaining time until evacuation. 
\begin{figure}[!ht]
                \centering
                \includegraphics[scale=0.5]{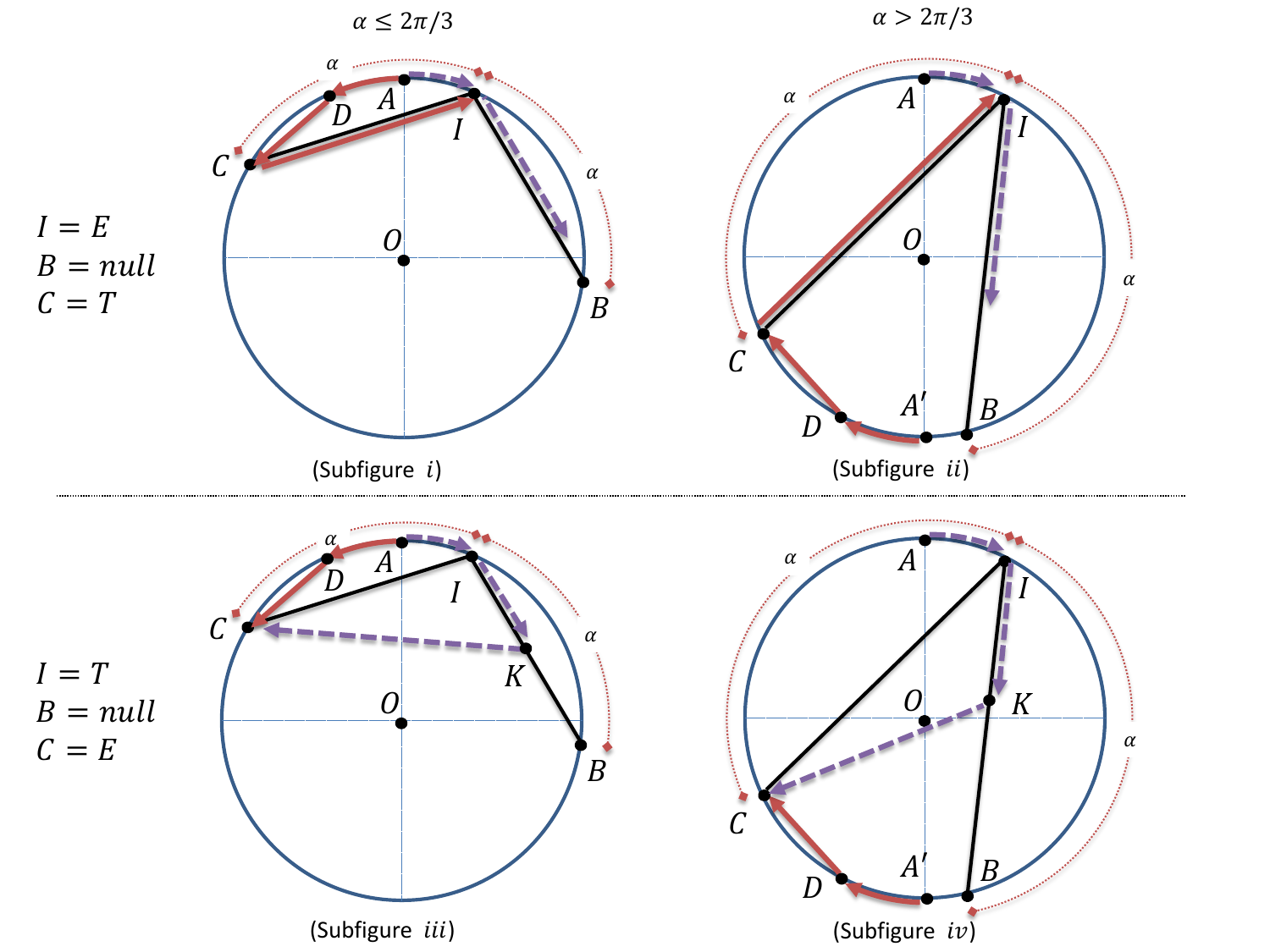}
                \caption{The performance of the wireless algorithm, when the starting point $A$ lies in the arc $\arcc{CI}$ of the two PoI. The trajectory of \ro{2} is depicted with the dotted purple curve, while the trajectory of \ro{1} with the solid red curve.}
                \label{fig: wireless-algo A in}
\end{figure}
At time $x$ the cases are as follows.

\begin{description}
\item[($I=E, B=null, C=T~\&~{\dist}\leq2\pi/3$):] Let \ro1 be at point $D$, \textit{i.e.} $\arcc{DA}=x$, see also Figure~\ref{fig: wireless-algo A in}i.  Then \ro1 moves along the chord $CD$, it locates the treasure and returns it to the exit $I$, with total cost
\begin{align*}
 \arcc{DA}+\barr{DC}+\barr{CI}
= & x+2\sinn{{\dist}/2-x}+2\sinn{{\dist}/2} \\
 \stackrel{(Lemma~\ref{lem:ti}\ref{1})} {\leq}
& \boundw.
\end{align*}
\item[($I=E, B=null, C=T~\&~{\dist}>2\pi/3$):] Let \ro1 be at point $D$, \textit{i.e.} $\arcc{DA'}=x$, see also Figure~\ref{fig: wireless-algo A in}ii. Then \ro1 moves along the chord $CD$, it locates the treasure and returns it to the exit $I$, with total cost
\begin{align*}
 \arcc{DA'}+\barr{DC}+\barr{CI}
\leq & x+2\sinn{\pi-\dist-x/2}+2\sinn{{\dist}/2} \\ \stackrel{(x\leq \pi-\dist)}{\leq} & \boundw.
\end{align*}

\item[($I=T, B=null, C=E~\&~{\dist}\leq2\pi/3$):] 
When \ro2 finds the treasure, it picks it up, and start moving along chord $IB$, see also Figure~\ref{fig: wireless-algo A in}iii. Meanwhile, \ro1 at time $x$ is at some point, say, $D$, and crosscuts through $CD$ to check the possible point $C$. When \ro1 visits $C$, \ro2 learns where the exit is, so starting from point, say, $K$, it moves along the line segment $KC$ and evacuates. Note that $K$ lies always in the the line segment $IB$, since $\barr{CD}\leq \barr{CI}=\barr{IB}$). The total cost then is 
\begin{align*}
&\arcc{AI}+\barr{IK}+\barr{KC}
=\arcc{AI}+\barr{CD}+\barr{KC} \\
&\leq \arcc{AI}+\barr{CD}+\max\{ \barr{CI}, \barr{CB} \} \\
&= x+2\sinn{{\dist}/2-x}+\max\{ 2\sinn{{\dist}/2},2\sinn{\dist} \} \\
&\stackrel{(Lemma~\ref{lem:ti}\ref{4})}{\leq} x+2\sinn{{\dist}/2-x}+ 2\sinn{\dist} \\
&\stackrel{(Lemma~\ref{lem:ti}\ref{1})}{\leq} \boundw.
\end{align*}

\item[($I=T, B=null, C=E~\&~{\dist}>2\pi/3$):] 
When \ro2 finds the treasure, it picks it up, and start moving along chord $IB$, see also Figure~\ref{fig: wireless-algo A in}iv.
Meanwhile, \ro1 at time $x$ is at some point, say, $D$, and crosscuts through $CD$ to check the possible point $C$. When \ro1 visits $C$, \ro2 learns where the exit is, so starting from point, say $K$, it moves along the line segment $KC$ and evacuates. Note that $K$ lies always in the the line segment $IB$, 
since 
\begin{align*}
\barr{CD} =&  2\sinn{\pi/2-{\dist}/2-x} \\
\stackrel{(Lemma~\ref{lem:ti}\ref{5})}{\leq}
& 2\sinn{{\dist}/2}
= \arcc{IB}.
\end{align*}
In return, the cost becomes 
\begin{align*}
& \arcc{AI}+\barr{IK}+\barr{KC} 
=\arcc{AI}+\barr{CD}+\barr{KC} \\
&\leq \arcc{AI}+\barr{CD}+\max\{ \barr{CI}, \barr{CB} \} \\
&= x+2\sinn{\pi/2-{\dist}/2-x}+\max\{ 2\sinn{{\dist}/2},2\sinn{\dist} \} \\
&\stackrel{(Lemma~\ref{lem:ti}\ref{4})}{\leq} x+2\sinn{\pi/2-{\dist}/2-x}+ 2\sinn{{\dist}/2} \\
&\stackrel{(Lemma~\ref{lem:ti}\ref{2})}{\leq} \boundw.
\end{align*}
\end{description}

\end{proof}

\begin{lemma}\label{lem: wireless Performance when A is out}
Let $A$ be the starting point of \ro2 which is the first to discover a PoI $I$. Let also the other PoI be at $B$, where $\arcc{IB}=\alpha$. If $A$ lies outside the arc $\arcc{IB}$, then the performance of Algorithm \ref{def: wireless algo} is $1+ \boundw$, for all $\dist\in [0,\pi]$. 
\end{lemma}

\begin{proof}
For the case analysis below, we rely on Figure~\ref{fig: wireless-algo A out}.
As before, robots spend time 1 to reach the periphery of the disk. Below we calculate the remaining time until evacuation. 
\begin{figure}[!ht]
                \centering
                \includegraphics[scale=0.50]{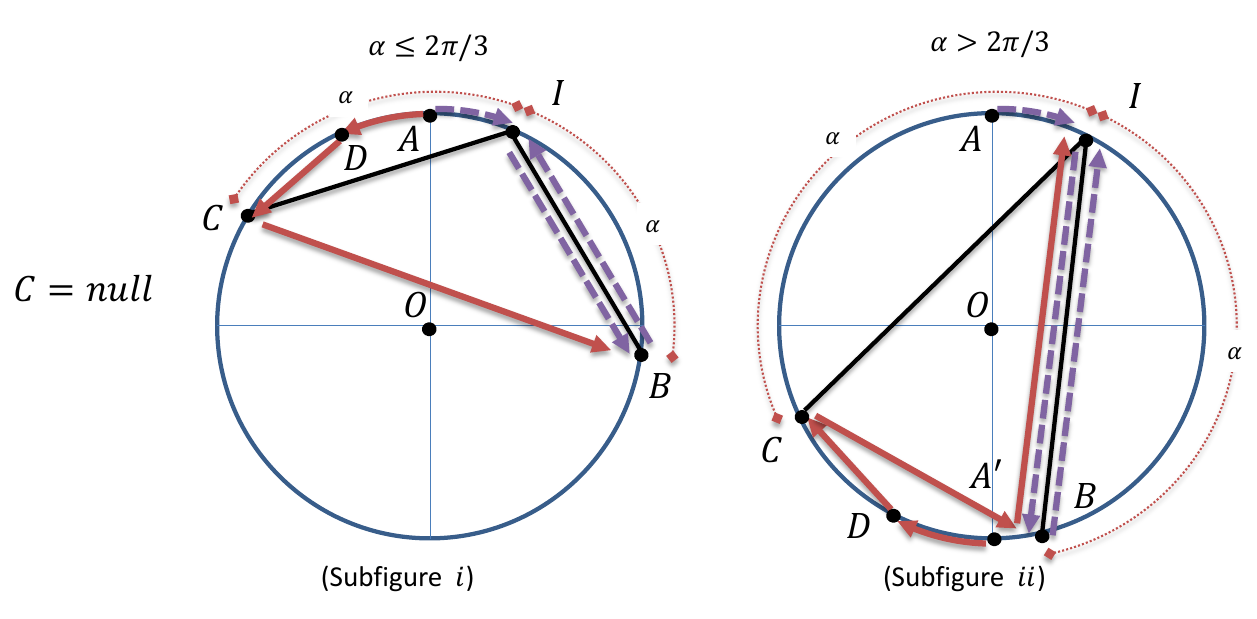}
                \caption{The performance of the wireless algorithm, when the starting point $A$ lies outside the arc $\arcc{IB}$ of the two PoI. The trajectory of \ro2 is depicted with the dotted purple curve, while the trajectory of \ro1 with the solid red curve.}
                \label{fig: wireless-algo A out}
\end{figure}
At time $x$ the cases we consider are as follows.
\begin{description}
\item[($C=null~\&~{\dist}\leq2\pi/3$):] 
After \ro2 discovers $I$ it will move along chord $IB$ to discover the other PoI, see also Figure~\ref{fig: wireless-algo A out}i. In particular, since $I$ is visited before $C$ we have $x\leq {\dist}/2$. 
If the treasure is in $B$, then the total cost would be 
\begin{align*}
\arcc{AI}+2\barr{IB} 
= & x+2\sinn{{\dist}/2} \\
\leq & {\dist}/2+4\sinn{{\dist}/2} \\
\leq & \boundw,
\end{align*}
(since ${\dist}\leq 2\pi/3$), while if the treasure is in $I$, then the cost would be $\pi-\dist+2\sinn{\dist/2}$.

\item[($C=null~\&~{\dist}>2\pi/3$):] 
After \ro2 discovers $I$ it will move along chord $IB$ to discover the other PoI, see also Figure~\ref{fig: wireless-algo A out}ii. In particular, since $I$ is visited before $C$ we have that $x\leq \pi-\dist$. 
If the treasure is in $B$, then the two robots are competing as to which will reach the treasure first. Even if \ro2 reaches the treasure first, the cost would be
$
\arcc{AI}+2\barr{IB}=x+2\sinn{{\dist}/2}\leq \boundw,
$
while if \ro1 reaches the treasure first, the total time will be even less than our promised upper bound. Finally, if the treasure is $I$, then the cost would be by $2\sinn{{\dist}/2}$ less than our promised upper bound. 
\end{description}
\end{proof}

It follows from Lemmata~\ref{lem: wireless Performance when A is in}, \ref{lem: wireless Performance when A is out} that, for all $\dist\in [0,\pi]$, the overall performance of Algorithm \ref{def: wireless algo} is no more than $1+\boundw$ concluding Theorem~\ref{thm: wireless upper bound}.

\section{Face-to-face Model}
\label{sec: f2f-model}

The main contribution of our work pertains to the face-to-face model and is summarized in the following theorem. 

\begin{theorem}\label{thm: non-wireless upper bound}
For every $\dist\in [0,\pi]$, problem \ten\ can be solved in time $1+\pi-\dist/2+3\sinn{\dist/2}$. 
\end{theorem}

Next we give the high-level intuition of the proposed evacuation-protocol, \textit{i.e.} Algorithm~\ref{def: nonwireless algo}, that proves the above theorem (more low level intuition, along with the formal description of the protocol appears in Section~\ref{sec: desciption f2f}). 

Denote by $\beta$ the upper bound provided by the theorem above. It should be intuitive that when the distance of the PoI $\dist$ tends to 0, there is no significant disadvantage due to lack of communication. And although the wireless evacuation-time might not be achievable, a protocol similar to the wireless case should be able to give efficient solutions. Indeed, our face-to-face protocol is a greedy algorithm when $\dist$ is not too big, \textit{i.e.} the two robots try independently to explore, locate the PoI and fetch the treasure to the exit without coordination (which is hindered anyways due to lack of communication). Following the worst case analysis, it is easy to observe that as long as $\dist$ does not exceed a special threshold, call it $\dist_0$ (which we define formally later and which is approximately $1.22353$), the evacuation time is $\beta$, and the analysis is tight. 

When $\dist$ exceeds the special threshold $\dist_0$, the lack of communication has a more significant impact on the evacuation time. To work around it, robots need to exchange information which is possible only if they meet. For this reason (and under some technical conditions), robots agree in advance to meet back in the center of the disk to exchange information about their findings, and then proceed with fetching the treasure to the exit. Practically, if the rendezvous is never realized, \textit{e.g.} only one robot reaches the center up to some time threshold, that should deduce that PoI are not located in certain parts of the disk, potentially revealing their actual location. In fact, this protocol works well, and achieves evacuation time $\beta$, as long as $\dist$ does not exceed a second threshold, which happens to be $2\pi/3$. 

The hardest case is when the two PoI are further than $2\pi/3$ apart. Intuitively, in such a case there is always uncertainty as to where the PoI are located, even when one of them is discovered. At the same time, the PoI, hence the robots, might be already far apart when some or both PoI are discovered. As such, meeting at the center of the disk to exchange information would be time consuming and induces evacuation time exceeding $\beta$. Our technical contribution pertains exactly to this case. Under some technical conditions, the treasure-finder might need to decide which of the two possible exit-locations to consider next. In this case, the treasure-holder follows a trajectory not towards one of the possible locations of the exit, rather a trajectory closer to that of its peer robot aiming for a rendezvous. The two trajectories are designed carefully so that the location of the exit is revealed no matter whether the rendezvous is realized or not.

\subsection{Algorithm \& Correctness}
\label{sec: desciption f2f}

In our main Algorithm \ref{def: nonwireless algo}, robots $R_1,R_2$ that start from the centre  of the circle, move together to an arbitrary point $A$ on the circle (which takes time 1). Then they start moving in opposing directions, say, counter-clockwise and clockwise respectively till they locate some PoI. 

In what follows we describe only the trajectory of \ro2 which is meant to be moving clock-wise (\ro1 performs the completely symmetric trajectory, and will start moving counter clock-wise). In particular all point references in the description of our algorithm, and its analysis, will be from the perspective of \ro2's trajectory which is assumed to be the robot that first visits either the exit or the treasure at position $I$. By $B,C,D$ we denote the points on the circle with $\arcc{DC}=\arcc{CI}=\arcc{IB}=\dist$ (see Figure~\ref{fig: general pic 3 d}). 
As before, and in what follows, $I\in \{E,T\}$ represents the position on the circle that is first discovered in the time horizon by any robot (in particular by \ro2), and that holds either the treasure or the exit. Finally, $O$ represents the centre  of the circle, which is also the starting point of the robots. 

\begin{figure}[!ht]
                \centering
                \includegraphics[scale=0.5]{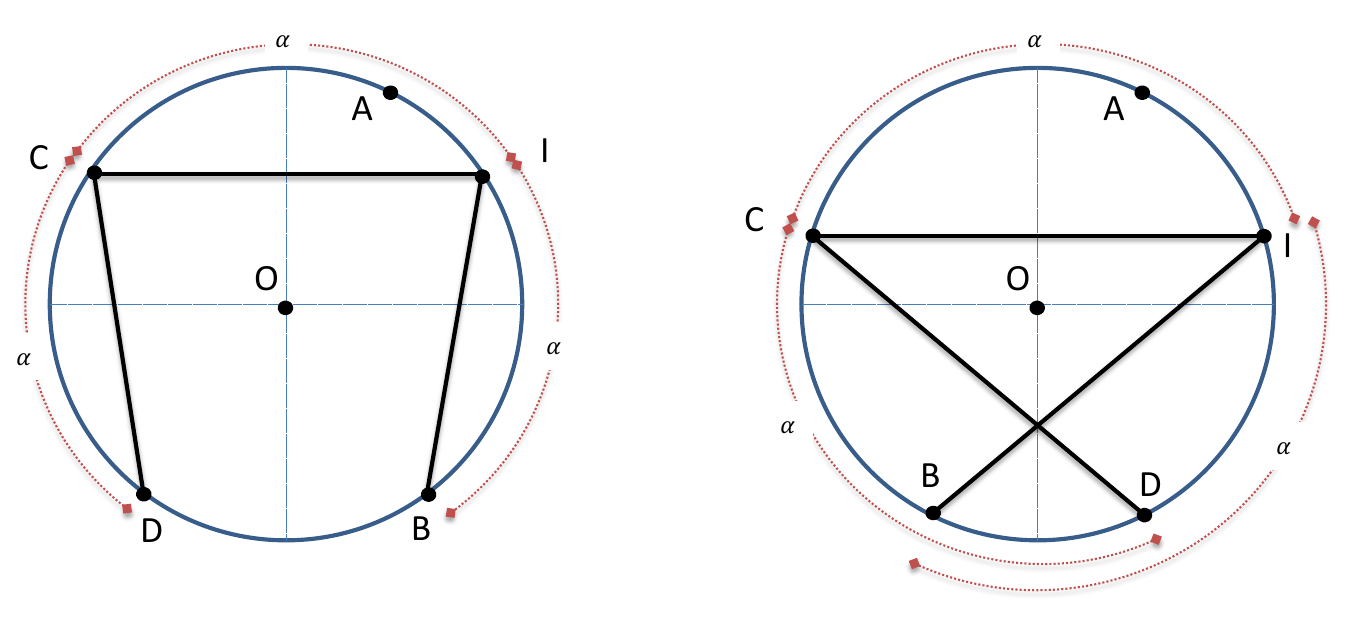}
                \caption{The points of interest from the perspective of \ro2, when $\dist\leq 2\pi/3$ on the left, and when $\dist\geq 2\pi/3$ on the right.}
                \label{fig: general pic 3 d}
\end{figure}


According to our algorithm, \ro2 starts moving from point $A$ till it reaches a PoI $I$ at time $x:=\arcc{AI}$. At this moment, our algorithm will decide to run one of the following subroutines with input $x$. These subroutines describe evacuation protocols, in which the treasure must be brought to the exit. Occasionally, the subroutines claim that robots evacuate (with the treasure) from points that is not clear that hold an exit. As we will prove correctness later, we comment on these cases by writing that ``correctness is pending''. 

\begin{description}
\item[\alg1(x)]
(Figure~\ref{fig: non-wireless-algo} i): If $I=T$, pick up the treasure and move to $B$ along the chord $IB$. If $B=E$ evacuate, else go to $C$ along the chord $BC$ and evacuate. \\
(Figure~\ref{fig: non-wireless-algo} ii): If $I=E$ move to $B$ along the chord $IB$. If $B=T$, pick up the treasure, and return to $I$ along the chord $BI$ and  evacuate. If $B=null$, then go to $C$ along the chord $BC$. If the treasure is found at $C$, pick it up, and move to $I$ along the chord $CI$ and evacuate (else abandon the process).

\item [\alg2(x)]
(Figure~\ref{fig: non-wireless-algo} iii): At the moment robots leave point $A$, set the timer to 0. \\
If $I=T$, pick up the treasure and go to the centre $O$ of the circle. Wait there till the time $t_0:=\max\{x, {\dist}-x+2\sinn{{\dist}/2}\}+1$. If \ro1 arrives at $O$ by time $t_0$, then go to $C$ and evacuate (correctness is pending). Else (if \ro1 does not arrive at $O$ by time $t_0$) go to $B$ and evacuate (correctness is pending). \\
(Figure~\ref{fig: non-wireless-algo} iv): If $I=E$,  move to $B$ along the chord $IB$. If $B=T$, pick up the treasure, and return to $I$ along the chord $BI$ and evacuate. If $B=null$, then go to the centre $O$ and halt.

\item[\alg3(x)]
(Figure~\ref{fig: non-wireless-algo} v): If $I=T$ pick up the treasure. If \ro1 is already at point $I$ go to $C$ and evacuate (correctness pending). If \ro1 is not at point $I$, then move along chord $ID$ for additional time 
$y:={\dist}/2-x+\sinn{{\dist}/2}+\sinn{\dist}$, 
and let $K$ be such that $\barr{IK}=y$. If \ro1 is at point $K$, then go to $B$ and evacuate (correctness is pending), else (if \ro1 is not at point $K$) go to $C$ and evacuate (correctness pending). \\
(Figure~\ref{fig: non-wireless-algo} vi): If $I=E$,  move to $B$ along the chord $IB$. If $B=T$, pick up the treasure, and return to $I$ along the chord $BI$ and evacuate. If $B=null$, then move along chord $BC$ 
until you hit $C$ (or you meet the other robot- whatever happens first) 
and halt at the current point, call it $K$. 
\end{description}

\begin{figure*}[!ht]
                \centering
                \includegraphics[scale=0.54]{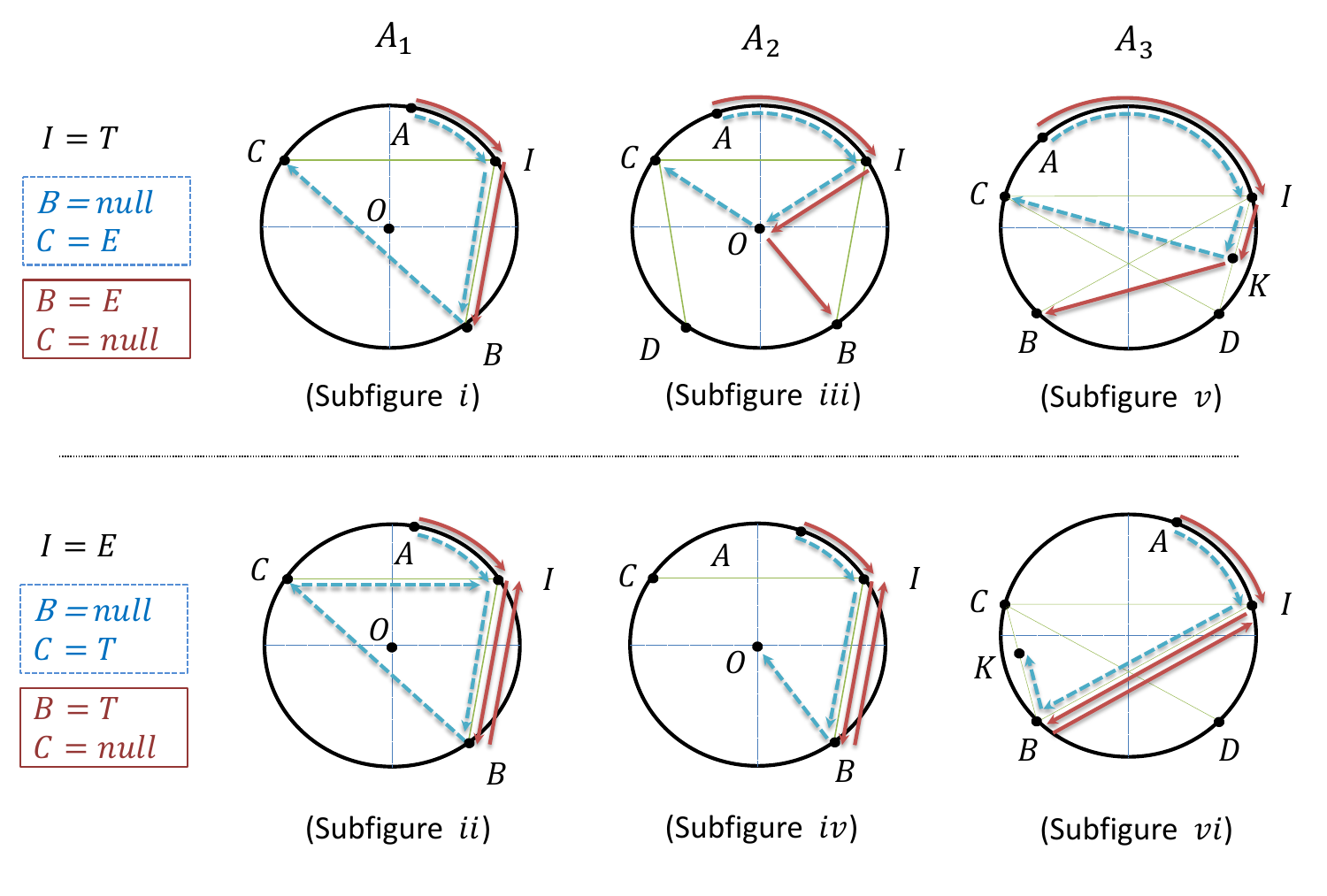}
                \caption{The non-wireless algorithm for two robots with performance $1+ \pi-{\dist}/2+3\sinn{{\dist}/2}$.}
                \label{fig: non-wireless-algo}
\end{figure*}

It is worthwhile discussing the intuition behind the subroutines above. First note that if a robot ever finds a treasure, it picks it up. The second important property is that
each robot simulates \alg1 either till it finds the treasure or till it fails to find the treasure after finding the exit. At a high level, \alg1 greedily tries to evacuate the treasure. This means that if the treasure is found first, then the robot tries successively the possible locations of the exit (using the shortest possible paths) and evacuates. If instead the exit is found, then it successively tries the (at most) two possibilities of the treasure location, and if the treasure is found, it returns it to the exit. 

\alg2 and \alg3 constitute our main technical contribution. Both algorithms are designed so that in some special cases, in which the exact locations of the PoI are not known, the two robots schedule some meeting points so that if the meeting (rendezvous) is realized or even if it is not, the treasure-holder can deduce the actual location of the exit. In other words, we make possible for the two robots to exchange information without meeting. 
Indeed after finding the treasure, in \alg2, \ro2 goes to the centre of the ring and waits some finite time till it makes some decision of where to move the treasure, while in \alg3, \ro2 moves along a carefully chosen (and non-intuitive) chord, and again for some finite time, till it makes a decision to move to a point on the ring. If instead the exit is found 
early, then the trajectories in \alg2, \alg3 are designed to support the other robot which might have found the treasure in case the latter does not follow \alg1. 

The next non-trivial and technical step would be to decide when to trigger the subroutines above. Of course, once this is determined, \textit{i.e.} once the trajectories are fixed, correctness and performance analysis is a matter of exhaustive analysis. 

We are ready to define our main non-wireless algorithm. We remind the reader that the description is for \ro2 that starts moving clockwise. \ro1 performs the symmetric trajectory by moving counter-clockwise. 

Our main algorithm uses parameter 
$\barr x(\dist) :=3{\dist}/2-\pi-\sinn{{\dist}/2}+2 \sinn{\dist},$
which we abbreviate by $\barr x$ whenever $\alpha$ is clear from the context. 
By Lemma~\ref{lem:ti}\ref{6}, $\dist_0\approx 1.22353$ is the unique root of $\barr x(\dist) =0$, while $\barr x$ is positive for all $\dist\in (\dist_0, \pi)$,  and negative for all $\dist \in [0,\dist_0)$. 
\begin{algorithm}
\caption{Non-Wireless Algorithm}
\label{def: nonwireless algo}
\begin{description}
\itemsep=-\parsep
\item
{\bf Step 1.} Starting from $A$, move clockwise until a PoI $I$ is found at time $x:=\arcc{AI}$.
\item
{\bf Step 2.} Proceed according to the following cases:
\begin{itemize}
\itemsep=-\parsep
\item If ${\dist}>2\pi/3$ and $I=T$ and ${\dist} > x\geq {\dist}-\barr x$, then run \alg3(x).
\item If ${\dist}>2\pi/3$ and $I=E$ and $x\leq \barr x$, then run \alg3(x).
\item If $\dist_0 \leq {\dist}\leq 2\pi/3$ and $I=T$ and ${\dist}> x\geq {\dist}-\barr x$, then run \alg2(x).
\item If $\dist_0 \leq {\dist}\leq 2\pi/3$ and $I=E$ and $x\leq \barr x$, then run \alg2(x).
\item In all other cases, run \alg1(x).
\end{itemize} 
\end{description}
\end{algorithm}
\begin{lemma}\label{lem: nonwireless correctness}
For every $\dist \in [0,\pi]$, Algorithm \ref{def: nonwireless algo} is correct, \textit{i.e.} a robot brings the treasure to the exit.
\end{lemma}

\begin{proof}
First, observe that the treasure is always picked up. Indeed, if the first PoI $I$ that is discovered (by any robot) is the treasure, then the claim is trivially true. If the first PoI $I$ found, say, by \ro2 is an exit, then \ro2 (in all subroutines) first tries the possible location $B$ for the treasure, and if it fails it tries location $C$ (in other words it always simulates \alg1 till it fails to find the treasure after finding the exit). Meanwhile \ro1 moves counter-clockwise on the ring, and sooner or later will reach $C$ or $B$. So at least one of the robots will reach the treasure first. In what follows, let \ro2 be the one who found first the treasure (and picks it up). We examine three cases. 

If \ro2 is following subroutine \alg1, then the treasure is brought to the exit. Indeed, in that case \ro2 expects no interaction from \ro1 and greedily tries to evacuate (see subcases i,ii in Figure~\ref{fig: non-wireless-algo}). 

If \ro2 is following subroutine \alg2, then it must be that $\dist_0\leq {\dist} \leq 2\pi/3$, and ${\dist}-\barr x \leq x< {\dist}$, and that it has not found any other PoI before (by Lemma~\ref{lem:ti}\ref{6} we have ${\dist}-\barr x< {\dist}$ and $\barr x> 0$, for all ${\dist}>\dist_0$). Figure~\ref{fig: non-wireless-algo} subcase iii depicts this scenario, where $I=T$. Note that from \ro2's perspective, the exit can be either in $B$ or in $C$, and \ro2 chooses to go to the center. This takes total time $x+1$. If the exit was at point $C$, then \ro1 would have found it in time ${\dist}-x\leq \barr x$ and that would make it to follow \alg2. So, \ro1 would first check point $D$ (where the treasure is not present), and that would make it to go to the centre arriving at time ${\dist}-x+2\sinn{{\dist}/2}+1$ (an illustration of this trajectory is shown in Figure~\ref{fig: non-wireless-algo} subcase iv, if \ro1 was moving clockwise). \ro2 is guaranteed to wait at the center till time $t_0$ (which is the maximum required time that takes each robot to reach the centre). In that case, \ro2 meets \ro1 at the center (because \ro1 did find the exit in $C$), and \ro2 correctly chooses $C$ as the evacuation point. Finally, if instead the exit was not in $C$, then \ro1 would not make it to the centre by time $t_0$. That can happen only if the exit is at point $B$, and once again \ro2 makes the right decision to evacuate from $B$.

In the last case, \ro2 is following subroutine \alg3, and so it must be that ${\dist} > 2\pi/3$, that ${\dist}-\barr x \leq x< {\dist}$, and that it has not found any other PoI before. Figure~\ref{fig: non-wireless-algo} subcase v depicts this scenario. Note that the exit could be either in $C$ or in $D$. 

If the exit is in $C$, then ${\dist}-x\leq \barr x$, and \ro1 would follow \alg3 too. This means, \ro1 would go to point $D$ (where there is no treasure), and that would make it travel along the chord $DI$ (an illustration of this trajectory is shown in Figure~\ref{fig: non-wireless-algo} subcase vi, if \ro1 was moving clockwise). If \ro1 reaches $I$, it waits there, and when \ro2 arrives in $I$, \ro2 makes the right decision to evacuate from $C$. Otherwise \ro1 does not reach $I$, and it moves up to a certain point on the chord $ID$ similarly to \ro2. Note that the meeting condition on a point $K$ on the chord, with $y=\barr{IK}$, would be that $\arcc{AI}+\barr{IK}=\arcc{CA}+\barr{CD}+(\barr{DI}-\barr{IK})$, which translates into $y=x+\sinn{{\dist}/2}+\sinn{\dist}-{\dist}/2$, \textit{i.e.} the exact segment of $ID$ that \ro2 traverses before it changes trajectory. The longest \ro2 could have traveled on  the chord $ID$ would be when $x={\dist}-\barr x$, but then $\barr{IK}$ would be equal to ${\dist} - \pi + 3 \sinn{\dist} \leq 2\sinn{\dist}=\barr{ID}$, for all ${\dist}>2\pi/3$. Therefore, the two robots meet indeed in somewhere on the chord $ID$. Note also that in this case, \ro2 makes the right decision and goes to point $C$ in order to evacuate. 

If instead the exit is in $B$, then again \ro2 travels till point $K$ (which is in the interior of the chord $ID$). But in this case, \ro1 will not meet \ro2 in point $K$ as it will not follow \alg3. Once again, \ro2 makes the right decision, and after arriving at $K$ it moves to point $B$ and evacuates. 
\end{proof}

\subsection{Algorithm Analysis}

In this section we prove that, for all $\dist\in [0,\pi]$, the evacuation time of Algorithm \ref{def: nonwireless algo} is no more than $1+\pi-{\dist}/2+3\sinn{{\dist}/2}$, concluding Theorem~\ref{thm: non-wireless upper bound}. In the analysis below we provide, whenever possible, supporting illustrations, which for convenience may depict special configurations. In the mathematical analysis we are careful not to make any assumptions for the configurations we are to analyze. 

It is immediate that when a robot finds the first PoI at time $x\geq \alpha$ after moving on the perimeter of the disk, then that robot can also deduce where the other PoI is located. In that sense, it is not surprising that, in this case, the trajectory of the robots and the associated cost analysis are simpler. 
\begin{lemma}\label{lem: Performance when x>=d}
Let $x$ be the time some robot is the first to reach a PoI $I \in\{E,T\}$ from the moment robots start moving in opposing directions. 
If $x \geq \dist$, then the performance of the algorithm is at most $1+\pi-{\dist}/2+3\sinn{{\dist}/2}$. Also, $x\geq \dist$ is impossible, if ${\dist}>2\pi/3$. 
\end{lemma}

\begin{proof}
Note that 1 is the time it takes both robots to reach a point, say $A$, on the ring. So we will tailor our analysis to the evacuation time from the moment robots start moving (in opposing directions) from point $A$. 

Let $x$ be the time after which \ro2 (without loss of generality) is the first to find a PoI $I \in \{E,T\}$. Let also $B$ be the other PoI $\{E,T\} \setminus I$. 
For \ro2 to reach first $I$, it must be the case that \ro1 does not have enough time to reach $B$, and hence $x\leq 2\pi-{\dist}-x$, that is $x\leq \pi-{\dist}/2$. Since also $x\geq \dist$, we conclude that ${\dist}\leq 2\pi/3$. 

Next we examine the following cases. For our analysis, the reader can use Figure~\ref{fig: general pic 3 d} as reference (although $A$ is depicted in the interior of the arc $CI$, we will not use that $\arcc{AI}\leq \dist$). 

\begin{description}
\item{Case 1 ($I=T$):} \ro2 picks up the treasure and moves along the chord $\barr{IB}=2\sinn{{\dist}/2}$. The worst case treasure-evacuation time then is 
$
\max_{{\dist}\leq x \leq \pi-{\dist}/2} \left\{ x+2\sinn{{\dist}/2} \right\}
=
\pi-{\dist}/2+2\sinn{{\dist}/2}.
$

\item{Case 2 ($I=E$):} According to the algorithm, \ro2 moves towards the treasure point $B$ along the chord $IB$, and reaches it in time $x+2\sinn{{\dist}/2}$. \ro1 moves counter-clock wise and will reach the position of the treasure in time $2\pi - {\dist} - x$. Whoever finds the treasure first will evacuate from the exit, paying additional time $2\sinn{{\dist}/2}$. Hence, the total cost can never exceed 
\begin{align*}
& \min \left\{ x+2\sinn{{\dist}/2}, 2\pi - {\dist} - x \right\} + 2\sinn{{\dist}/2}  \\
{\leq} & \boundnw. ~~~~~~\mbox{(by Lemma~\ref{lem:ti}\ref{7})}
\end{align*}
\end{description}
Observe that in both cases, the cost of the algorithm is as promised. 
\end{proof}

By Lemma~\ref{lem: Performance when x>=d} we can focus on the (much more interesting) case that \ro2, which is the first robot that finds a PoI, arrives at $I$ at time $x:=\arcc{AI}< {\dist}$. A reference for the analysis below is Figure~\ref{fig: non-wireless-algo} which is accurately depicting point $A$ at most $\alpha$ arc-distance away from $I$. For the sake of better exposition, we examine next the cases ${\dist}\leq 2\pi/3$ and $\alpha \geq 2\pi/3$ separately. Note that in the former case robots may run subroutines \alg1 or \alg2, while in the latter case robots may run subroutines \alg1 or \alg3. For the lemma below, the reader may consult Figures~\ref{fig: general pic 3 d} and \ref{fig: non-wireless-algo}.

\begin{lemma}\label{lem: Performance when x<dist}
Let $x$ be the time some robot is the first to reach a PoI $I \in\{E,T\}$ from the moment robots start moving in opposing directions. 
If $x < {\dist}$, then the performance of the algorithm is at most $1+\pi-{\dist}/2+3\sinn{{\dist}/2}$, for all $\dist\in [0,\pi]$. 
\end{lemma}

\begin{proof}
As before, we omit in the analysis below the time cost 1, \textit{i.e.} the time robots need to reach the periphery of the disk. We examine the following cases for \ro2, which is the robot that finds $I$.
\begin{description}
\item[($I=T, B=E, C=null$):] 
If \ro2 runs \alg1, then it must be that $x\leq {\dist}-\barr x$, so the cost is $x+2\sinn{{\dist}/2} \leq {\dist}-\barr x + 2\sinn{{\dist}/2} 
\leq \pi-{\dist}/2+3\sinn{{\dist}/2}$ (see Figure~\ref{fig: non-wireless-algo} i). 

If \ro2 runs \alg2, then it must be that ${\dist}-\barr x \leq x < {\dist}$ and ${\dist}\leq 2\pi/3$, and the robot goes to the centre in order to learn where the exit is (see Figure~\ref{fig: non-wireless-algo} iii). Independently of where the exit is, and by Lemma~\ref{lem: nonwireless correctness}, \ro2 makes the right decision and evacuates in time $1+\max_{{\dist}-\barr x \leq x < {\dist}}\{x,{\dist}-x+2\sinn{{\dist}/2} \}+1 \leq \max\{\alpha, \barr x + 2\sinn{{\dist}/2}\} +2$ 
which, by Lemma~\ref{lem:ti}\ref{8}, is at most $\boundnw$, for all ${\dist}\leq 2\pi/3$.
Note that the analysis of this case is valid, even if $I=T$ is not the first PoI that is discovered, and it is from the perspective of the robot that finds the treasure. 

If \ro2 runs \alg3, then it must be that ${\dist}-\barr x \leq x < {\dist}$ and ${\dist}> 2\pi/3$. Then the trajectory of \ro2 is as in Figure \ref{fig: non-wireless-algo} v, and the exit is found correctly due to Lemma~\ref{lem: nonwireless correctness}. For the sake of the exposition, we will do the worst case analysis for both cases $B=E$ and $C=E$ now (\textit{i.e.} we only insist in that $I=T$ and that \ro2 runs \alg3). 

The total time for the combined cases is 
$\arcc{AI}+\barr{IK}+\max\{\barr{KB}, \barr{KC}\},$
where $\barr{IK}=y$ (see definition of \alg3). Since as we have proved, $K$ lies always in chord $ID$, and since $\arcc{DB}=3{\dist}-2\pi$ we have that 
\begin{align*}
\barr{BK} 
 \leq & \max\{\barr{BI}, \barr{BD}\} \\
\leq &\max\{2\sinn{{\dist}/2}, 2\sinn{3{\dist}/2-\pi}\} \\
\leq & 2\sinn{{\dist}/2}.
\end{align*}
 We also have that $\barr{KC} \leq \barr{CI}=2\sinn{{\dist}/2}$. So the cost becomes no more than 
\begin{align*}
& x+y+2\sinn{{\dist}/2} 
= {\dist}/2+3\sinn{{\dist}/2}+\sinn{\dist} \\ 
{\leq} & \pi-{\dist}/2+3\sinn{{\dist}/2},
~~~~~\mbox{(by Lemma~\ref{lem:ti}\ref{9})}
\end{align*}
for all $\dist\in [0,\pi]$. 

\item[($I=T, B=null, C=E$):] Since $I$ is found first, we must have $x\leq {\dist}/2$, hence both robots run \alg1, see Figure \ref{fig: non-wireless-algo} i. Robot \ro2 that finds the treasure will evacuate in time no more than 
$
x+2\sinn{{\dist}/2}+2\sinn{\dist}\leq {\dist}/2+2\sinn{{\dist}/2}+2\sinn{\dist} < \pi-{\dist}/2+3\sinn{{\dist}/2}$, for all $\dist\in [0,\pi]$. 

\item[($I=E, B=T, C=null$):] If \ro2 is the first to find the treasure, then this case is depicted in Figure \ref{fig: non-wireless-algo} i. This happens exactly when $x+2\sinn{{\dist}/2}\leq 2\pi-x-\alpha$, so that the total evacuation time is $x+4\sinn{{\dist}/2}\leq  \pi-{\dist}/2+3\sinn{{\dist}/2}$, for all $\dist\in [0,\pi]$. 

Otherwise $x>\pi-{\dist}/2-\sinn{{\dist}/2}$, and \ro1 is the robot that reaches the treasure first. If \ro1 decides to run \alg1, then the cost would be $2\pi-x-{\dist}+2\sinn{{\dist}/2} <  \pi-{\dist}/2+3\sinn{{\dist}/2}$, for all $\dist\in [0,\pi]$. Finally, if \ro1 decides to run \alg2 or \alg3, then we have already made the analysis in case $I=T, B=E, C=null$ above. 

\item[($I=E, B=null, C=T$):] Note that in all cases, both robots will run the same subroutine. In particular, if robots run either \alg2 or \alg3, then we have already done the analysis in case $I=T, B=E, C=null$ above. 

Finally, if both robots run \alg1, it must be either because ${\dist}\leq \dist_0$, or because $x\geq \barr x$, while the cost is always ${\dist}-x+2\sinn{{\dist}/2}+2\sinn{\dist}$ (the case is depicted in Figure \ref{fig: non-wireless-algo} ii, with reverse direction). 
If ${\dist}\leq \dist_0$, then the evacuation cost would be at most ${\dist}+2\sinn{{\dist}/2}+2\sinn{\dist}$ which by Lemma~\ref{lem:ti}\ref{3} is at most $\boundnw$, for all $\dist\in [0,\dist_0]$. If $x\geq \barr x$, then the cost would be at most ${\dist}-\barr x+2\sinn{{\dist}/2}+2\sinn{\dist} = \pi-{\dist}/2+3\sinn{{\dist}/2}$.
\end{description}
\end{proof}

Note that Lemmata~\ref{lem: Performance when x>=d},
\ref{lem: Performance when x<dist} imply that the performance of Algorithm \ref{def: nonwireless algo} is, in the worst case, no more than $1+\pi-{\dist}/2+3\sinn{{\dist}/2}$, concluding also Theorem~\ref{thm: non-wireless upper bound}.

\section{Lower Bounds}
\label{sec: lb f2f}

We conclude the study of treasure evacuation with 2 robots by providing the following lower bound pertaining to distributed systems under the face-to-face communication model.
\begin{theorem}   \label{thm:2lb}
For problem \ten, any algorithm needs at least time $1+\pi/3+4\sinn{\dist/2}$\ if\ $0\leq \dist\leq 2\pi/3$, or $1+\pi/3+
2\sinn{\dist}+2\sinn{\dist/2}$\ if\ $2\pi/3 \leq \dist\leq \pi$.
\end{theorem}

For the proof, we invoke an adversary (not necessarily the
most potent one),
who waits for as long as there are
three points $A,B,C$ with $AB=BC=\dist$ on the periphery such that at most one of them has been visited by a robot. Then depending on the moves of the robots decides where to place the PoI. 


\begin{proof}[Proof of Theorem~\ref{thm:2lb}]
Since the robots start from the center, they'll need time $1$ to reach the periphery. The adversary (not necessarily the
most potent one, but with this weaker adversary we still get a (weaker) lower bound) will wait for as long as there are
three points $A,B,C$ with $AB=BC=\dist$ on the periphery such that at most one of them has been visited by a robot.
Observe that this will be true for as long as less than $2\pi/3$ of the periphery has been explored; this will be 
done by the 2 robots after time at least $(2\pi/3)/2=\pi/3$. Hence, after time at least $1+\pi/3$ there are such points $A,B,C$
with only one of these points visited by a robot. For convenience, we assume that robot 1 is the first to visit a point
at time $t$ and robot 2 visits a different point next at time $t+\varepsilon$ (if this doesn't happen, then the optimal algorithm
would be behaving like the case of only one robot, which is clearly suboptimal for the adversary moves below). It will
be apparent below that the lower bound becomes weaker for $\varepsilon=0$, so that's what we will assume from now on.
We distinguish the following cases:

\noindent {\bf Case 1 (Robot 1 at $A$, Robot 2 at $C$):} If the adversary places $T\rightarrow B, E\rightarrow A \text{ or } C$, then recovery needs
at least time $4\sinn{\dist/2}$ (if robot 1 or 2 respectively evacuates $T$ by itself). If it places $T\rightarrow A, E\rightarrow C$, 
then recovery needs at least time $2\sinn{\dist}$ (a robot evacuates $T$ by traversing $AC$). Any other placement of $T,E$
by the adversary gives either the same or a worse (lower) bound, and, therefore, it's discarded. It is clear that, in this case, 
the adversary goes with the first option, for a lower bound of $4\sinn{\dist/2}$.  

\noindent{\bf Case 2 (Robot 1 at $A$, Robot 2 at $B$):} If the adversary places $T\rightarrow C, E\rightarrow A$, then recovery needs
at least time $\min\{2\sinn{\dist/2}+2\sinn{\dist}, 4\sinn{\dist}\}$ (if robot 2 or 1 respectively evacuates $T$ by itself). 
If it places $T\rightarrow C, E\rightarrow B$, then recovery needs at least time $\min\{4\sinn{\dist/2}, 2\sinn{\dist}+2\sinn{\dist/2}\}$ 
(if robot 2 or 1 respectively evacuates $T$ by itself). Any other placement of $T,E$
by the adversary gives either the same or a worse (lower) bound, and, therefore, it's discarded. It is clear that, in this case, 
the adversary goes with the option that maximizes the lower bound, for a lower bound of 
$$ \max\left\{
\begin{tabular}{ll}
$\min\{2\sinn{\dist/2}+2\sinn{\dist}, 4\sinn{\dist}\},$ \\
$\min\{4\sinn{\dist/2}, 2\sinn{\dist}+2\sinn{\dist/2}\}$
\end{tabular}
\right\}.
$$ 
By taking the minimum of Cases 1,2 above, the lower bound of the theorem follows.
\end{proof}

\section{Conclusion}
\label{sec:conclusion}

In this paper we introduced a new problem on {\em searching and fetching} which we called \emph{treasure-evacuation} from a unit disk. We studied two online variants of treasure-evacuation with two robots, based on different communication models. 
The main point of our approach was to propose distributed
algorithms 
by a collaborative team of robots.
Our main results demonstrate how robot communication capabilities affect the treasure evacuation time by contrasting face-to-face (information can be shared only if robots meet) and wireless (information is shared at any time) communication. 

There are several open problems in addition to sharpening our bounds, and in particular to giving lower bounds that would effectively separate the two communication models. 
Other variations of our problem include the consideration of different
1) number of robots, 
2) geometric domains (discrete or continuous), 
3) robot starting positions, 
4) number of hidden objects,
5) communication models, 
6) robots' speeds, 
7) a priori knowledge of the topology or partial information about the targets (\textit{e.g.} a bound on the distance of the hidden items or no information at all), etc. 
In each case, the challenging task is to establish either tight bounds, or to separate closely related problems, \textit{e.g.} the problem in which either the exact distance vs a bound on the distance of the hidden items is known. 

When it comes to searching with multiple robots, our 2-robot algorithms can be easily extended to the $n$-robot case. Assuming that $n$ is even (otherwise we ignore one robot) we split the robots into pairs, defining points in intervals of length $4\pi/n$ on the cycle, assigning each pair of robots to each such poin. Then, we let them run the corresponding 2-robot algorithm. Would that strategy be improvable? 
We anticipate that nearly optimal algorithms for small number of robots, \textit{e.g.} for $n=3,4$, or any other variation of problem we consider will require new and significantly different algorithmic ideas than those we propose here, still in the same spirit. 

\acknowledgements
\label{sec:ack}
The authors would like to thank the anonymous reviewers for proposing a number of suggestions that resulted in various improvements. 

\nocite{*}
\bibliographystyle{abbrvnat}
\bibliography{refs}
\label{sec:biblio}

\appendix

\section{Trigonometric Inequalities}
\begin{lemma}\label{lem:ti}
\begin{enumerate}[a)]
\item \label{6} 
There exists some $\dist_0 \in (0,\pi)$ such that $3{\dist}/2-\pi-\sinn{{\dist}/2}+2 \sinn{\dist}$ is positive for all $\dist\in (\dist_0, \pi)$, and negative for all $\dist \in [0,\dist_0)$. In particular,  $\dist_0\approx 1.22353$. 

\item \label{7} 
$\min \left\{ x+2\sinn{{\dist}/2}, 2\pi - {\dist} - x \right\} + 2\sinn{{\dist}/2} 
\leq \boundnw, ~\forall \dist\in[0,\pi]$. 

\item \label{8} 
$\max\{\alpha, \barr x + 2\sinn{{\dist}/2}\} +2  \leq \boundnw$ for all $\dist\in [0, 2\pi/3]$.

\item \label{9} 
${\dist}+\sinn{\dist} \leq \pi$ for all $\dist\in [0,\pi]$.

\ignore{
\item \label{10} 
${\dist}-\sinn{{\dist}/2}+2\sinn{\dist} \leq \pi$ for all $\dist\in [0,\pi]$.
}


\item \label{3} 
${\dist}/2+2\sinn{\dist}\leq \pi-{\dist}+2\sinn{{\dist}/2}, ~\forall \dist\in [0,2\pi/3]$.  .

\item \label{5} 
$\max_{0\leq x\leq \pi-\dist}\{
\sinn{\pi/2-{\dist}/2-x}
\} \leq \sinn{{\dist}/2}, ~\forall \dist\in [2\pi/3, \pi]$

\item \label{1} 
$\max_{0\leq x\leq {\dist}/2}\{x+2\sinn{{\dist}/2-x}\} +2\sinn{\dist} \leq \boundw, ~\forall \dist\in [0,2\pi/3]$. 

\item \label{2} 
$\max_{0\leq x\leq \pi-\dist}\{x+2\sinn{\pi/2-{\dist}/2-x}\}
\leq \pi-{\dist}+2\sinn{{\dist}/2}, ~\forall \dist\in [2\pi/3, \pi]$
\ignore{
\begin{align*}
&\max_{0\leq x\leq \pi-\dist}\{x+2\sinn{\pi/2-{\dist}/2-x}\} \\
\leq & \pi-{\dist}+2\sinn{{\dist}/2}, ~\forall \dist\in [2\pi/3, \pi]
\end{align*}
}
\item \label{4} $\sinn{\dist}\leq \sinn{{\dist}/2},~\forall \dist\in [0,2\pi/3]$, and \\ $\sinn{\dist}\geq \sinn{{\dist}/2},~\forall \dist\in [2\pi/3,\pi]$.
\end{enumerate}
\end{lemma}

\begin{description}
\item[Proof of~\ref{lem:ti}\ref{6}] We observe that 
\begin{align*}
\frac{\partial}{\partial \alpha}\barr x(\dist)
=& \frac{\partial}{\partial \alpha}\left( 3{\dist}/2-\pi-\sinn{{\dist}/2}+2 \sinn{\dist}\right)\\
= & 3/2-\coss{\dist}+\coss{{\dist}/2}.
\end{align*}
Observe that the above quantity remains positive for ${\dist}<2\pi/3$, while it is negative for ${\dist}>2\pi/3$. Since $\barr x (0)<0$ and $\barr x(2\pi/3)>0$, it follows that there is a unique root $\dist_0\in (0,2\pi/3)$ (which numerically can be estimated to $\dist_0\approx 1.22353$). Finally, we see that $\barr x(\pi) = \pi-1>0$, so $\barr x (\dist)$ remains positive for $\dist\in [2\pi/3, \pi]$. 

\item[Proof of~\ref{lem:ti}\ref{7}] We observe that 
$\min \left\{ x+2\sinn{{\dist}/2}, 2\pi - {\dist} - x \right\}$ attains its maximum when $x+2\sinn{{\dist}/2}= 2\pi - {\dist} - x$, in which case its value becomes $\boundnw$.

\item[Proof of~\ref{lem:ti}\ref{8}] 
First we claim that $3{\dist}/2-2\sinn{{\dist}/2} \leq \pi-2$ for ${\dist}\leq 2\pi/3$. This is because $\frac{\partial}{\partial \alpha}\left( 3{\dist}/2+2\sinn{{\dist}/2}\right) = 3/2+\coss{{\dist}/2}>0$, hence $3{\dist}/2-2\sinn{{\dist}/2} \leq \pi - \sqrt{3}/2 \leq \pi-2$. This claim immediately shows that 
${\dist}+2  \leq \boundnw$ for all $\dist\in [0, 2\pi/3]$.

Now we show that $\barr x + 2\sinn{{\dist}/2}+2 \leq \boundnw$ for all $\dist\in [0, 2\pi/3]$. For this it suffices to check that ${\dist}+\sinn{\dist}-\sinn{{\dist}/2} \leq \pi-1$. To that end we see that 
$\frac{\partial}{\partial \alpha}\left( {\dist}+\sinn{\dist}-\sinn{{\dist}/2} \right) = 1 + \coss{\dist}-\coss{{\dist}/2}/2 \geq 0$ for all ${\dist}\leq 2\pi/3$. Hence 
${\dist}+\sinn{\dist}-\sinn{{\dist}/2} \leq 2\pi/3
\leq 2\pi/3 + \sqrt{3}/2-\sqrt{3}/2 \leq \pi-1$ as wanted.

\item[Proof of~\ref{lem:ti}\ref{9}] We see that 
$\frac{\partial}{\partial \alpha}\left( {\dist}+\sinn{\dist} \right) = 1+\coss{\dist}\geq 0$, for all $\dist\in [0,\pi]$. hence, ${\dist}+\sinn{\dist} \leq \pi+\sinn{\pi}=\pi$.

\ignore{
\item[Proof of~\ref{lem:ti}\ref{10}] See Figure~\ref{fig: trigineq-noproof}.
\begin{figure}[!ht]
                \centering
                \includegraphics[scale=0.5]{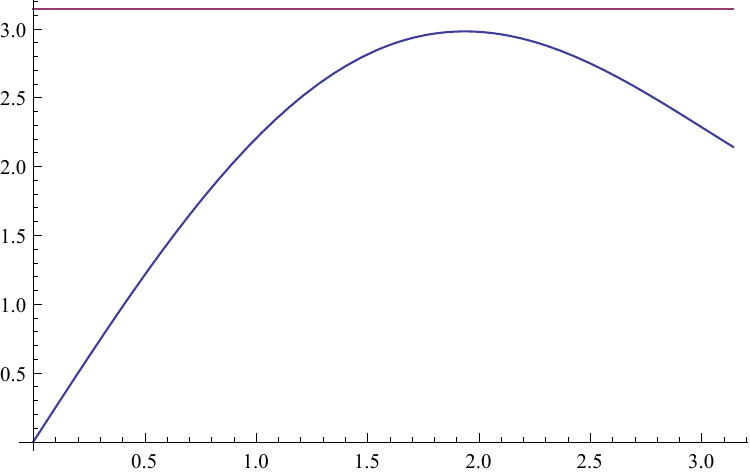}
                \caption{The function ${\dist}-\sinn{{\dist}/2}+2\sinn{\dist}$ compared to $\pi$ for all $\dist\in [0,\pi]$.}
                \label{fig: trigineq-noproof}
\end{figure}
}


\item[Proof of~\ref{lem:ti}\ref{3}] We observe that 
\begin{align*}
& \frac{\partial}{\partial \alpha}\left( 3{\dist}/2+2\sinn{\dist}-2\sinn{{\dist}/2}\right) \\
= &
3/2+2\coss{\dist}-\coss{{\dist}/2}.
\end{align*}
From the monotonicity of cosine in $[0,2\pi/3]$, we see that the above derivative preserves non negative sign when $\dist\in [0,2\pi/3]$. Hence, the maximum of 
$$3{\dist}/2+2\sinn{\dist}-2\sinn{{\dist}/2} \leq \pi$$
is attained when $\dist=2\pi/3$, and its value is $\pi$ as wanted. 

\item[Proof of~\ref{lem:ti}\ref{5}] We have that
\begin{align*}
& \max_{0\leq x\leq \pi-\dist}\{\sinn{\pi/2-{\dist}/2-x}\} \\
= & \max_{0\leq x\leq \pi-\dist}\{\coss{{\dist}/2+x}\} \\
\leq & \coss{{\dist}/2},
\end{align*}
since cosine is monotonically decreasing in $[0,\pi]$. But also for all $\dist\in [2\pi/3, \pi]$ we have that $\coss{{\dist}/2}\leq \sinn{{\dist}/2}$, concluding what we need.

\item[Proof of~\ref{lem:ti}\ref{1}] We have that 
\begin{align*}
& \max_{0\leq x\leq {\dist}/2}\{x+2\sinn{{\dist}/2-x} +2\sinn{\dist} \} \\
\leq & {\dist}/2+\sinn{{\dist}/2} +\sinn{\dist}
\end{align*}
where the first inequality is true due to the monotonicity of $x,\sinn{{\dist}/2-x}$ w.r.t. $x\leq {\dist}/2$ and for all $\dist\in [0,2\pi/3]$, and the last inequality since again ${\dist}\leq 2\pi/3$. The claim now follows from Lemma~\ref{lem:ti}\ref{3}.

\item[Proof of~\ref{lem:ti}\ref{2}] 
Follows immediately since $x\leq \pi-\dist$, and by Lemma~\ref{lem:ti}\ref{5}.

\item[Proof of~\ref{lem:ti}\ref{4}] Observe that $\sinn{{\dist}/2}-\sinn{\dist}$ is convex in $\dist\in [0,2\pi/3]$, so it attains its maximum either at $\dist=0$ or at $\dist =2\pi/3$. In both cases, its value is 0. Also, $\sinn{{\dist}/2}-\sinn{\dist}$ is monotonically increasing for ${\dist}>2\pi/3$, implying what was promised. 
\end{description}

\end{document}